\documentclass[11pt,twoside]{article}
\usepackage[letterpaper,margin=1in]{geometry}

\usepackage{amssymb}
\usepackage{graphicx}
\usepackage{verbatim}
\usepackage{amsmath,textcomp,amssymb,geometry,graphicx,enumerate}
\usepackage[boxruled,linesnumbered,vlined]{algorithm2e}
\usepackage{hyperref}

\usepackage[T1]{fontenc}

\usepackage{lmodern}
\usepackage{amsthm}

\usepackage{dsfont}
\usepackage{xargs}
\usepackage{todonotes}
\usepackage{tikz}

\usepackage{amssymb}
\usepackage{wrapfig}
\usepackage{mathrsfs}
\usepackage{enumerate}
\usepackage{caption}
\usepackage{microtype}
\usepackage{soul}
\usepackage{wrapfig}
\usepackage[charter]{mathdesign}
\usepackage{graphicx}
\usepackage{mathrsfs}
\usepackage{enumerate}
\usepackage{caption}
\usepackage{nicefrac}

\newif\iflong
\longtrue

\newtheorem{theorem}{Theorem}
\newtheorem{lemma}[theorem]{Lemma}

\newtheorem{definition}[theorem]{Definition}

\newtheorem{observation}[theorem]{Observation}
\newtheorem{conjecture}[theorem]{Conjecture}

\newcommand{\BFS}{\textsc{Bfs}}

\newcommand{\RR}{\mathcal{L}}

\newcommand{\BB}{\mathcal{B}}
\newcommand{\CC}{\mathcal{C}}
\newcommand{\TT}{\mathcal{T}}
\newcommand{\PP}{\mathcal{P}}

\newcommand{\OM}{\omega}
\newcommand{\SSS}{\mathcal{S}}
\newcommand{\TL}{\tilde O}
\newcommand{\NS}{\sqrt {n\sigma}}
\newcommand{\NBS}{\sqrt {\frac{n}{\sigma}}}
\newcommand{\SI}{\sigma}

\newcommand{\LCA}{\textsc{Lca}}
\newcommand{\SSR}{\textsc{Ssrp}}
\newcommand{\MSR}{\textsc{Msrp}}
\newcommand{\BMM}{\textsc{Bmm}}
\DeclareMathOperator{\E}{\mathbb{E}}

\newcommand{\SUF}{\textsc{Suffix}}

\newcommand{\MTC}{\textsc{Mtc}}
\newcommand{\QUE}{\textsc{Query}}

\newcommand{\Cross}{$\mathbin{\tikz [x=1.4ex,y=1.4ex,line
width=.2ex] \draw (0,0) -- (1,1) (0,1) -- (1,0);}$}

\usepackage{tikz}

\author{
  Manoj Gupta,\\
  IIT Gandhinagar,\\
  \texttt{gmanoj@iitgn.ac.in}
  \and
  Rahul Jain\textsuperscript{$\dagger$} \footnote{\textsuperscript{$\dagger$} The work was done  when the author was a student at IIT Gandhinagar}
,\\
  Goldman Sachs, Bangalore\\
  \texttt{rahul.e.jain@gs.com}
  \and
  Nitiksha Modi,\\
  IIT Gandhinagar,\\
  \texttt{nitiksha.modi@iitgn.ac.in}
}
\begin{document}

\title{Multiple Source  Replacement Path Problem }
\maketitle

\begin{abstract}
One of the classical line of work in graph algorithms   has been the Replacement Path Problem:
given a graph $G$, $s$ and $t$, find   shortest paths from $s$ to $t$ avoiding each
edge $e$ on the shortest path from $s$ to $t$.
These paths are called replacement paths in literature.
 For an undirected and unweighted graph, (Malik, Mittal, and Gupta, Operation Research Letters, 1989) and (Hershberger and Suri, FOCS 2001)
designed an algorithm that solves the replacement path problem in $\TL(m+n)$ time\footnote{$\TL$ notation hides poly$\log n$ factor.}.
  It is natural to ask whether we can generalize the replacement path problem: {\em can we find all replacement paths from a source $s$ to all vertices in $G$?} This problem is called the Single Source Replacement Path Problem.

Recently (Chechik and Cohen, SODA 2019) designed a randomized combinatorial algorithm that solves the Single Source Replacement Path Problem in $\TL(m\sqrt n\ + n^2)$ time.
One of the questions left unanswered by their work is the case when there are many sources, not one. When there are $n$ sources,  the combinatorial algorithm of (Bernstein and Karger, STOC 2009) can be used to find all pair replacement path in $\TL(mn + n^3)$ time.  However, there is no result known for any general $\SI$. Thus, the problem we study is defined as follows: given a set of $\SI$ sources, we want to find the replacement path from these sources to all vertices in $G$. We give a randomized combinatorial algorithm for this problem that takes $\TL(m\NS +\ \SI n^2)$ time. This result generalizes both results known for this problem. Our algorithm is much  different and arguably simpler than (Chechik and Cohen, SODA 2019). Like them, we show  a matching conditional lower bound using the
Boolean Matrix Multiplication conjecture.

\end{abstract}


\section{Introduction}
One of the classical line of work in graph algorithms is the
replacement path problem.  The general setting for this problem is as follows:
we are given a graph $G$ and two vertices $s$ and $t$. We want to output the
length of all shortest paths avoiding edges on the $st$ path. Note that
we just want to output the length of these paths -- not the path itself. These
paths are called replacement paths in the literature.

Replacement paths were first investigated due to their relation with auction theory,
where they were used to compute the Vickrey pricing of edges owned by selfish agents
 \cite{Hershberger2001,NisanR01}. One can also generalize the replacement path problem and ask to output $k$
--not just one -- replacement paths \cite{Roditty2012a,Bernstein2010,GotthilfL2009}.
This is also called the $k$-replacement path problem.

The replacement path problem has been extensively studied.
There are algorithms \cite{MalikMG89,NardelliPW03,Hershberger2001} that compute replacement paths in an undirected
and unweighted graph
in
 $\TL(m + n)$ time. For directed, unweighted graphs, Roditty and
Zwick \cite{Roditty2012a}
designed an algorithm that  finds all
replacement paths in $O(m
\sqrt n)$ time.
For the $k$-shortest path problem, Roditty
\cite{Roditty2007}
presented
an algorithm with an approximation ratio
$\nicefrac{3}{2} $, and the running time  $O(k(m\sqrt n+n^{3/2}
\log n))$.
 Bernstein \cite{Bernstein2010} improved the
above result
to get an approximation factor of $(1+\epsilon)$
and running
time $O(km/\epsilon)$.  The same paper also
gives an improved
algorithm for the approximate $st$ replacement
path problem.

Let us now generalize the replacement path problem, where we want to find replacement paths from $s$ to every other vertex in the graph. We define this problem
formally.

{\bf Single Source Replacement Path} ($\SSR$): {\it Given a graph $G$ and a source $s$, design an algorithm that can find
the length of shortest paths from $s$ to every vertex $t \in V$ avoiding each edge on $st$ path.}

Grandoni and Williams \cite{GrandoniW12} were the first to study the $\SSR$ problem in a directed graph. For a directed graph having weights in the range $[1, M]$, they designed an algorithm for the $\SSR$ problem that takes $O(M n^{\OM})$ time (where $\OM$ is
the matrix multiplication exponent \cite{Gall14a,Williams12}). They also presented results when
there are negative weights on the edges. Recently, Chechik and Cohen \cite{ChechikC19}
designed an algorithm that solves the single
source
replacement path problem in $\TL(m \sqrt
n + n^2)$ time in an undirected graph. They also show that this
time
is nearly optimal using the conditional lower
bound of Boolean matrix multiplication. Other related work includes \cite{EmekYPR2010,Eppstein94,RodittyZ11,Williams11,WilliamsW18}.

  We now generalize the $\SSR$ problem when there are multiple sources, we call this the Multiple
Source Replacement Path Problem.

{\bf Multiple Source Replacement Path Problem} ($\MSR$): {\it  Given a graph $G$ and a set of sources $\SSS$,
design an algorithm that finds length of all replacement paths from $s$ to $t$ where $s \in \SSS$ ($|\SSS| = \SI$) and
$t \in V$.}

To the best of our knowledge, there is only one work that designs a combinatorial algorithm to solve the above problems.
If there are $n$ sources, then the work of Bernstein and Karger \cite{Bernstein2009} can be used to
find the replacement path for any pair of vertices in $\TL(mn + n^3)$ time. In their work \cite{Bernstein2009},
Bernstein and Karger built a distance oracle of size $\TL(n^2)$ that
can answer the following query: $\QUE(x,y,e)$: find the length of the replacement path
from $x$ to $ y$ avoiding $e$ where $x,y \in V$. They designed an algorithm that
builds this oracle in $\TL(mn)$ time and answers queries in $O(1)$ time. Once this oracle is built, one can just query
this oracle and find all replacement paths between any pair of vertices. This takes $O(n^3)$
time. Thus, when there are $n$ sources, the  $\MSR$ problem can
be solved in $\TL(mn + n^3)$ time.

The main open question left behind by these two works is the case
when there are $|\SSS| = \SI$ sources. In this paper, we solve this question by showing the
following theorem:

\begin{theorem}
  There is a randomized combinatorial algorithm that solves the $\MSR$ problem in $\TL(m\NS + \SI n^2)$ time\footnote{Note that the second term in the running
time is needed as there may be $\Omega(\SI n^2)$
terms to output.}.
\end{theorem}

The reader can see that there are two combinatorial results for the above problem. If $\SI = 1$, then we have the result of Chechik and Cohen \cite{ChechikC19}. And when $\SI= n$, we have the result of Bernstein and Karger \cite{Bernstein2009}. Our result generalizes both these results.
Additionally, we extend the conditional lower first presented in \cite{ChechikC19} by giving a combinatorial reduction
from Boolean Matrix Multiplication (\BMM) to $\MSR$ problem.

\begin{theorem}
\label{thm:second}
For a combinatorial algorithm $\MSR(n,m)$ with runtime of $T(n,m)$, there is a combinatorial algorithm for $\BMM(n,m)$  problem with runtime of $O(\NBS T(O(n),O(m)))$.
\end{theorem}

\subsection{Related Work}
Bernstein and Karger \cite{Bernstein2009} solved the $\MSR$ problem when $\SI = n$. As mentioned
previously, their aim was to build a single edge fault tolerant distance oracle of size $\TL(n^2)$
with a query time of $O(1)$. Demestrescu et al. \cite{Demetrescu2008} were the first to
design this distance oracle of size $\TL(n^2)$ and query time $O(1)$. However,
they did not specify the running time of their algorithm. Bernstein and Karger \cite{Bernstein2009}
answered this question. Other related work in this area are \cite{Bernstein2008,Chechik2010,ChechikCFK17,Duan2009}

The fault-tolerant distance oracle itself can be generalized when there is a single source
or $\SI $ many sources. Bilo et al. \cite{BiloCG0PP18} designed a distance oracle of size $\TL(\SI^{1/2} n^{3/2})$
that can answer single fault queries in $O(\NS)$ time. Gupta and Singh \cite{GuptaS18} reduced the query
time of this oracle to $\TL(1)$.

Other related problems include the fault tolerant subgraph problem.
The aim of this problem is to find a subgraph of $G$ such that the shortest path from
$s  \in \SSS$ is preserved in the subgraph after any edge deletion.
Parter and Peleg \cite{Parter2013} designed an algorithm to compute single fault tolerant subgraph with
$ O(n^{3/2})$ edges. They also showed that their result can be easily extended to multiple sources
with $O(\SI^{1/2}n^{3/2})$ edges. This result was later extended
to dual fault by Parter [16] with $ O(n^{5/3}) $ edges. Gupta and Khan \cite{GuptaK17} extended
the above result to multiple sources with $O(\SI^{1/3}n^{5/3})$ edges.
All the above results are optimal due to a result by Parter \cite{Parter2015} which states
that a multiple source $k$ fault tolerant  subgraph requires $O(\SI^{\frac{1}{k+1}}n^{2-\frac{1}{k+1}})$
edges. There is only one positive result known
for a general $\SI$, Bodwin et al. \cite{BodwinGPW17} showed the existence of a $k$ fault tolerant
subgraph of size $O(k \SI^{1/2^k} n^{2-1/2^k} )$.

\section{Previous Approach : Chechik and Cohen \cite{ChechikC19}}
\setlength\intextsep{0pt}
\begin{wrapfigure}[12]{r}{.30\textwidth}
\centering
\begin{tikzpicture}[scale=1.5]

\definecolor{dgreen}{rgb}{0.0, 0.5, 0.0}

\begin{scope}[xshift=0cm]
\coordinate (s) at (0,2);
\coordinate (v) at (0,0);
\coordinate (c) at (0,.5);
\coordinate (b) at (0,3);
\coordinate (a) at (0,1.5);
\coordinate (d) at (-0.3,1);
\coordinate (d1) at (-0.3,1.1);
\coordinate (e1) at (0,1);
\coordinate (e2) at (-0.2,.65);


\draw[thick](v)--(s);
\node[above] at (s){$s$};
\node[below] at (v){$t$};
\node[right] at (a){$a$};
\node[right] at (c){$c$};


\draw[blue,thick] (a) to[out=170,in=170] node[pos=0.5,left]
{\scriptsize  $\SUF(P)$}  (c);
\node at (e1){$\times$};
\node[right] at (e1){$e$};

\end{scope}

\end{tikzpicture}

\caption{ $\SUF(P)$ starts with the  blue path  from $a$. It merges back to $st$ path at $c$ and continues   till $t$.
}

\label{fig:suffix}
\end{wrapfigure}
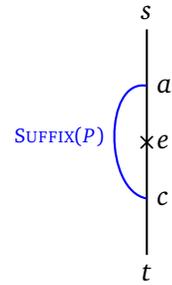

Before we dive into our approach, let us look at the previous approach to the problem. To this end, let us first define few terms.
\begin{enumerate}
\item Let $st$ denote the shortest path from $s$
to $t$ in $G$.

\item  Let $P$ be the shortest
replacement path
from $s$ to $t$
avoiding an edge $e$ on the $st$ path.
$\SUF(P)$ denotes
the suffix
of $P$ from the point it leaves the original
$st$ path.

\item {\em Landmark vertices:}
Let  $\RR$ be a set obtained by sampling each vertex in $G$
with a probability of $\frac{1}{\sqrt n}$.

\end{enumerate}
One can show that $\RR$ contains $\TL(\sqrt
n)$ vertices
 with high probability. We can also show
that if $\SUF(P)$ contains $\TL(\sqrt n)$
vertices, then with a high probability, $\SUF(P)$
contains a  landmark vertex.

We now describe the result of Chechik and Cohen \cite{ChechikC19} in detail. In  \cite{ChechikC19}, the authors solved the $\SSR$ problem using
the following important observation from \cite{AfekBKCM02}: ``For any replacement path from $s$ to $t$ avoiding $e$, there exists a vertex $u$ such that the replacement path  can be
broken into two shortest paths in $G$ (1) $su$ path and (2) $ut$ path''. Note that the result is non-trivial as these two paths are shortest paths in
the original graph -- that is none of these paths pass through $e$ in $G$.
If we are able to find the vertex $u$, then we can easily find the replacement path.
Unfortunately, finding $u$ is not easy (in the running time we want to achieve).

We now try to explain how Chechik and Cohen \cite{ChechikC19} overcome this
problem using landmark vertices.
Chechik
and Cohen \cite{ChechikC19} showed the following:   if $\SUF(P)$ is sufficiently long, say $\TL(\sqrt n)$, then it  contains two landmark vertices $u$ and $v$ such that $P$ can be
broken into three parts (1) $su$ path (2) $uv$ {\em path} and (3) $vt$ path. Again all these paths
are shortest paths in $G$ and do not pass through the edge to be avoided.
Since the number of landmark vertices is only $\TL(\sqrt n)$, finding $u$ and $v$
becomes slightly easy. Once we have the above result,
we can use the following simple algorithm (See Algorithm \ref{alg:che}) to find the shortest path from $s$ to $t$ avoiding $e$ (assuming that $\SUF(P)$ is sufficiently long).

\begin{algorithm}

$d(s,t,e) \leftarrow \infty$; \tcp{$d(s,t,e)$ will be equal to $|P|$ at the end of the algorithm.}

\For{each pair $u,v \in \RR$}
{ \If{$su, uv$ and $vt$ avoids  $e$ }
   {

        $d(s,t,e) \leftarrow \min\{d(s,t,e),
  |su| + |uv| + |vt|\}$;
   }
}

return $d(s,t,e)$
\caption{Algorithm for finding a replacement
path $P$  from $s$ to $t$ avoiding $e$ assuming that $\SUF(P)$ is long}
\label{alg:che}
\end{algorithm}

The reader can check that the time taken by Algorithm \ref{alg:che} for a fixed $(t,e)$ pair
 is $\TL(n)$ (assuming that checking if $su, uv$ and $vt$ contain $e$ takes $O(1)$ time).
Since there are $O(n^2)$ such pairs, the algorithm takes $\TL(n^3)$ time. However, remember that Chechik
and Cohen \cite{ChechikC19} solve the $\SSR$ problem in $\TL(m\sqrt n + n^2)$ time.
Thus, we need to reduce this running time from $\TL(n^3)$ to  $\TL(n^2)$.
The above simple algorithm is the heart of their paper (``the double pivot case'') and reducing the time to $\TL(n^2)$
require some more modification to the simple algorithm -- we donot describe this modification as it is technically heavy and
orthogonal to the approach used by us. We refer the interested reader to  \cite{ChechikC19} for details.

\section{Our Approach}
We now describe our approach for the case when there is only one source, that is $\SI = 1$. In between our explanation, we will also point out the differences between \cite{ChechikC19}
and our result.

We will use the classical result of \cite{MalikMG89,Hershberger2001,NardelliPW03}, that  can find the
replacement path from $s$ to $t$ in $\TL(m+n)$ time. An inefficient algorithm will
then be to run this algorithm for each $s,t$ pair where $t \in V$, giving a running time
of $\TL(mn)$. Since we want a better running time, we restrict
the use of the result of \cite{MalikMG89, Hershberger2001,NardelliPW03} to only landmark vertices $\RR$. Thus, we find
the replacement path from $s$ to vertices in  $\RR$. Since there are
$\TL(\sqrt n)$ landmark vertices, the total time taken is $\TL(m\sqrt n + n\sqrt n)$.
We now use the set $\RR$ to find replacement paths for all other vertices. We first note the first difference between \cite{ChechikC19} and our result, the use of landmark vertices is completely different from that in \cite{ChechikC19}.
We use the classical result to find the replacement path between $s$ to all
landmark vertices, unlike \cite{ChechikC19} where the properties of landmark vertices are used (in Algorithm \ref{alg:che}).

We now describe our strategy to find a replacement path $P$ from
$s$ to $t$ avoiding $e$ if $\SUF(P)$ is sufficiently long.
We claim  the following results (which we show using Lemma \ref{lem:rexists}):
If $\SUF(P)$ is of length $\TL(\sqrt n)$, then
with a high probability a landmark vertex, say $v$, will lie on $\SUF(P)$
such that $vt$ path does not contain $e$. Note that this result is similar to the one
used by Chechik and Cohen \cite{ChechikC19} -- here we are  arguing about $vt$ path only.
Thus, our algorithm to find replacement path (whose suffix is sufficiently long) is as follows:

\begin{algorithm}

$d(s,t,e) \leftarrow \infty$; \tcp{$d(s,t,e)$ will be equal to $|P|$ at the end of the algorithm.}

\For{each $v \in \RR$}
{ \If{$vt$ avoids  $e$ }
   {
        Let $p$ be the length of the replacement path from $s$ to $v$ avoiding $e$ \;

        $d(s,t,e) \leftarrow \min\{d(s,t,e),
  p + |vt|\}$;
   }
}

return $d(s,t,e)$
\caption{Algorithm for finding a replacement
path $P$  from $s$ to $t$ avoiding $e$ assuming that $\SUF(P)$ is long}
\label{alg:2}
\end{algorithm}
In the above algorithm, we use the fact that we have already found replacement paths from $s$ to all
the landmark vertices.
Even though we know that some landmark vertex lies on $\SUF(P)$, we do not know which one.
Thus, we have to scan all of $\RR$ to find the vertex which lies in $\SUF(P)$.
This process itself takes $\TL(\sqrt n)$ time. This implies that it would take $\TL(n \sqrt n )$
time to find all replacement paths of $t$, which in turns implies an $\TL(n^2\sqrt n)$
running time. Since we cannot afford such a running time, we use the following scaling trick.
This scaling trick is the second difference between \cite{ChechikC19} and our paper.
This trick greatly simplifies the algorithm as well as the analysis
in our paper.

We look at an edge $e$ at a distance of $[2^k, 2^{k+1}]$ from $t$ on $st$ path
(assume that $2^k \ge \sqrt n \log n$).
One can argue that the suffix of the replacement path, say $P$,
avoiding $e$ will have length $\ge 2^k$. This is because $e$ itself is at a
distance $\ge 2^k$ from $t$.
Since $|\SUF(P)| > \TL(\sqrt n)$, we know that one of our landmark vertex, say $v$, lies on $\SUF(P)$.
If we can find $v$ then we can use Algorithm \ref{alg:2} to find the length of $P$.
At the same time, we don't have enough time to look at all the landmark vertices
to find $v$. Here comes our main idea.
Since $|\SUF(P)| \ge 2^k$, we can choose a smaller set of landmark vertices $\RR_k$
of size $\TL(\frac{n}{2^k})$ (a set where each vertex is sampled with probability $\frac{1}{2^k}$).
We show (using Lemma \ref{lem:rexists}) that
there exists a landmark vertex $v \in \RR_k$ such that $v$ lies on $\SUF(P)$
and $vt$ path avoids $e$.
Since, we have already found all replacement paths for vertices in $\RR_k$
(using \cite{MalikMG89,Hershberger2001,NardelliPW03}),
we can use $\RR_k$ (instead of $\RR$) to find replacement paths to
$t$ when the edges are in the range $[2^k,2^{k+1}]$.

Thus, we use the same algorithm in Algorithm \ref{alg:2}, but use
the landmark set $\RR_k$ for edges that are at a distance of $[2^k, 2^{k+1}]$
from $t$.
As the value of $k$
increases the number of vertices in $\RR_k$
decreases. The end effect is that we take the same amount of time to process edges
 in the range $[2^k,2^{k+1}], [2^{k+1},2^{k+2}], [2^{k+2},2^{k+3}], \dots ,[n/2,n]$ from $t$.
In general, the reader will see that it will take $\TL(n)$ time to process all the edges in
any given range. This would imply a running time of $\TL(n)$ for finding all
replacement paths of $t$ whose suffix is long. This approach reduces the running time from $\TL(m\sqrt n + n^2\sqrt n)$
to $\TL(m\sqrt n + n^2)$.

Note that we cannot use the above approach for edges which are
near to $t$ in $st$ path (that is, those edges whose distance from $t$ is $\le \sqrt n \log n$)
or those replacement paths that have short suffixes.
It turns
out that dealing with these replacement paths is relatively easy. Even
Chechik and Cohen \cite{ChechikC19} have a simple approach for these paths.
In our paper, we design another algorithm
to deal with these replacement paths. This completes the description
of our algorithm for the single source case.

The above approach gives us a simple algorithm for finding replacement paths
from a single source to all vertices. We then extend our result to multiple sources.
Ideally, we would have liked to use the result of \cite{MalikMG89, Hershberger2001,NardelliPW03}
to find all replacement paths from all sources to vertices in $\RR$.
However, this does not give us the required running time.
To overcome this barrier, we show that we can adapt the result of
Bernstein and Karger \cite{Bernstein2009}
 to find all replacement paths between all sources and vertices in $\RR$ (in the required running time).
This completes the overview of our approach.

\section{Notation}
We use the following notation throughout the paper:

\begin{itemize}


\item Unless stated otherwise, $uv$ will denote the shortest path between
the vertex $u$ and $v$ in the graph $G$. $|uv|$ denotes the length of this shortest path.

\item Let $e = (u,v)$ be an edge on $st$ path such that $u$ is closer to $s$ than $v$. We will abuse notation and use $se$ to denote $su$ path and $et$ to denote $vt$ path.

\item $st \diamond e$ is the shortest path from $s$ to $t$ avoiding $e$. $|st \diamond e|$ is the length of this shortest path.

\item $uv + vy$ denotes the concatenation of two paths, one ending at $v$ and other starting at $v$.

\item Let $P$ be a path from $s$ to $t$, not necessarily the shortest path. A sub-path $uv$ on $P$
will be denoted by $P[u,v]$.

\item The shortest path tree of a vertex $v \in V$ is denoted by $\TT_v$. The shortest path tree can be built by performing Breadth First Search (\BFS) algorithm from $v$.

\item $d(s, t, e)$: In our algorithm, we want to find the shortest replacement path from each source to each vertex. $d(s,t,e)$ is the length of the replacement path (from $s$ to $t$ avoiding $e$) calculated by our algorithm. We will normally initialize $d(s,t,e) = \infty$ and prove that at the end of the algorithm $d(s,t,e) = |st ~\diamond~ e|$.
Also, $d(s,t)$ will denote the length of the shortest $st$ path. Formally, $d(s,t) = d(s,t, \emptyset) = |st|$.

\item The term {\em with a high probability} means with a probability  $\ge 1 - \frac{1}{n^c}$ where $c \ge 1$.

\end{itemize}

\section{Preliminaries}

We first sample a set of random vertices which we call as landmark vertices.

\begin{definition}
\label{def:landmark} (Landmark vertices)
Let  $\RR_{k}$ be a set of vertices sampled randomly from
$G$
with a probability of $\frac{4}{2^k}\sqrt{\frac{\SI}{n}}$
where $0\le k \le\log \NS$. Let $\RR = \cup_{k=0}^{\log \NS}
\RR_k$.
Along with these vertices, $\RR$  also contains all source
nodes.
\end{definition}

 The following lemma bound the number of vertices in $\RR$.

\begin{lemma}
\label{lem:sizeR}
The size of $\RR_k$ is $\TL(\frac{\NS}{2^k})$ with a very high
probability. Thus, the size of $\RR$ is $ \TL(\NS)$ with
a very high probability.
\end{lemma}
\begin{proof}
  Let $X_k$ be a random varaible denoting the size of $\RR_k$.
  The expected size of $X_k$ is, $\E[X_k] = \frac{4\NS}{2^k}$.
  Using Chernoff's bound, we know that $P[X_k \ge (1+ \delta)\E[X_k]] \le e^{-\frac{\delta \E[X]}{3}}$
  where $\delta \ge 1$.
  Putting $\delta = \log n$, we get
  $P[X_k \ge (1+ \log n)\frac{4\NS}{2^k}] \le e^{-\frac{4\log n \NS}{3 \times 2^k}}$.
  Since $2^k$ can at most be $\NS$, we get $P[X_k \ge (1+ \log n)\frac{4\NS}{2^k}] \le e^{-\frac{4 \log n}{3}} = n^{-4/3}$.
  Using union bound, the probability that the size of $\RR$ is $ \ge \TL(\NS)$ is $\le (\log \NS) \times n^{-4/3} \le n^{-1}$
  (where the last ineqality is true for a high enough value of $n$).
\end{proof}

Let us first build some elementary data-structure that
will help us in our algorithm. Using Breadth First Search (\BFS) algorithm,
we can find the shortest
path from $s$ to all other vertices in $G$ in $
O(m+n)$
time. We will assume that at the end of $\BFS$ algorithm,
we will find the distance from $s$ to every other vertex
in $G$, that is $d(s,v) = |sv|$.  Also, assume that we obtain the shortest
path tree of s, that is $\TT_s$, as the
output of $\BFS$  algorithm. We store  $d(s,v)$
in a hash-table for efficient retrieval. To this end, we
use the following data structure:
\begin{lemma}
\label{lem:cuckoo}\text{(Pagh and Rodler \cite{PaghR04})}
There exists a randomized hash-table with constant look-up
time in the worst case and constant insertion time in expectation.
\end{lemma}

For each landmark vertex $r$, we find the shortest path
from $r$
to every other vertex in $G$. This can again be done using
$\BFS$ algorithm and the total running time is $\tilde
O((m+n) \NS)$ as the number of landmark vertices  is $\tilde
O(\NS)$. We  store the length of the shortest path from
$r$ to
every other vertex $v$, that is $d(r,v)$, in a hash-table.

Using the result of \cite{MalikMG89, Hershberger2001, NardelliPW03}, we
know that all  replacement
paths from $s$ to a landmark vertex $r$ can  be found in $\tilde
O(m+n)$
time. If there is only one source, that is $\SI =1$, then
 we can use this result to find all replacement paths
between the single source $s$ and all landmark vertices.
 For a single source case, since  the number of landmark
vertices
is $\TL(\sqrt n)$, the
time taken to find all  replacement paths is $\tilde
O((m+n)\sqrt
n)$.

However, when there are many sources, the above
strategy gives a running time of $\TL((m+n)\SI\NS)$ where
the second and third multiplicand represent the number of
sources and the number landmark vertices. We cannot afford
such a huge running time.
So, we adapt the result of Bernstein and Karger \cite{Bernstein2009}
and show that it can be used to find all replacement paths from all sources to all landmark vertices in $\TL(m\NS+\SI
n^2)$ time. Given this result (which we will show in Section
\ref{sec:general}),
for each $s \in \SSS$ and \ $r \in
\RR$, we store $d(s,r,e)$ in a hash-table for each $e \in
sr$ path.

Lastly, we used the following classical result to compute
least common ancestors quickly:

\begin{lemma} \text{(See \cite{BenderF00} and its references)}
Given any tree $\TT_v$ (on $n$ vertices) rooted at $v$, we can
build a data-structure
of size $O(n)$ in $O(n)$ time which can find least common
ancestor \textit{i.e.}
$\LCA(x, y)$ where $x, y \in \TT_v$ in $O(1)$ time.
\end{lemma}

Remember that we want to find  all replacement paths
from $s$ to $t$ for each  $s \in \SSS$ and $t \in V$. Fix
a  $s \in \SSS$ and $t \in V$.
We partition the edges on the  $st$ path into two sets,
{\em far} and {\em near}.

\begin{itemize}
\item ($k$-Far Edges) Edges which are at a distance $[2^{k+1}
\NBS \log n, 2^{k+2} \NBS \log n]$ away from $t$ on $st$
path (where $0 \le k\le \log \NS)$.
\item (Near Edges) Edges  which are at a distance $<2 \NBS
\log n$ away from $t$ on $st$ path.\end{itemize}


\section{Far Edges}
\label{sec:far}
Fix a source $s \in \SSS$ and $t \in V$. Assume that we are trying to find a replacement path
for a $k$-far edge $e$ on $st$ path.
Since the replacement path avoids $e$,
it has to diverge from the $st$ path before $e$.
We now again look at the suffix of a replacement path and describe its properties.
\begin{definition} Let $P$ be the shortest replacement path
from $s$ to $t$
avoiding an edge $e$ on the $st$ path. Then, $\SUF(P)$ denotes
the suffix
of $P$ from the point it leaves the original $st$ path.

\end{definition}

Since $P$ avoids $e$ on $st$ path, it has to diverge from
this path before edge
$e$. $\SUF(P)$ is the sub path of $P$ that starts from this
diverging vertex. We now make an important observation about the
length of $\SUF(P)$.

\begin{observation}
\label{obs:sufflong}
If $P$ is the shortest replacement path from $s$ to $t$ avoiding a
$k$-far edge $e$ on $st$ path, then $|\SUF(P)| >2^{k+1} \NBS \log n$.
\end{observation}

The above observation holds because of the following simple argument: as $\SUF(P)$ start before $e$ on $st$ path, its length
should be $\ge$ length of $et$ path.
We now claim that there exists a vertex of $\RR_{k}$ on $\SUF(P)$. This can be shown easily
using elementary probability.

\begin{lemma}
\label{lem:rexists}
Let $\PP$ be the set of all replacement paths from $s \in \SSS$ to $t \in V$ that avoid a far edge. Given any path $P \in
\PP$ such that $P$ avoids a $k$-far  edge $e$ on $st_{}$ path, with a high probability there exists a vertex
 $r \in \RR_{k}$ on $\SUF(P)$ such that the distance of $r$
to $t$ on $\SUF(P)$ is $\le 2^k\NBS \log n$.
\end{lemma}

\begin{proof}
Fix a $P \in \PP$. Since $P$ avoids a $k$-far edge on $st$ path,
by Observation \ref{obs:sufflong}, $|\SUF(P)| > 2^{k+1} \NBS \log n$.
Let $X_P$ be the event that there does not exist a vertex of $\RR_{k}$
at a distance $\le 2^k\NBS \log n$ from $t$ on $\SUF(P)$.
Then the probability that $X_P$ occurs is
$P[X_p] = (1-\frac{4}{2^k}\sqrt{\frac{\SI}{ n}})^{2^{k}\NBS \log n} \le \frac{1}{n^4}$.
Note that the size of $\PP$ is $\le n^2\SI \le n^3$. Thus, the probability
that $X_P$ occurs for any $P \in \PP$ is $P[\cup_{P \in \PP} X_P] \le \frac{1}{n}$.
\end{proof}

We now use a simple algorithm (See Algorithm \ref{alg:que}) to find the shortest
replacement path from $s$ to $t$ avoiding a $k$-far edge $e$.

\begin{algorithm}
$d(s,t,e) \leftarrow \infty$\;

\For{each $r \in \RR_{k}$}
{

   \If{$d(r,t) \le 2^k \NBS \log n$}
   {

        $d(s,t,e) \leftarrow \min\{d(s,t,e), d(s,r,e)+d(r,t)\}$;
   }
}

return $d(s,t,e)$
\caption{Algorithm for finding a replacement path for a $k$-far edge $e$}
\label{alg:que}
\end{algorithm}

 Let $P$ be the replacement path from $s$ to $t$ avoiding a $k$-far edge $e$. Using
Lemma \ref{lem:rexists}, we know that there exists a landmark vertex $r \in
\SUF(P)$ ($r \in \RR_k$) such that the distance of $r$ to $t$ on $\SUF(P)
\le 2^{k}\NBS \log n$. Thus, the shortest path from $r$ to $t$, that
is $rt$, has length $\le 2^{k }\NBS \log n$. We first claim
that this path cannot pass through $e$. This is due to the
fact that $e$ is a $k$-far edge and the shortest path from $e$
to $t$ is $ \ge 2^{k+1} \NBS \log n$. Given
such a $r \in \RR_{k}$, finding the replacement path becomes
easy, it is  $d(s,r,e)+d(r,t)$. We have already
calculated both these terms in the preprocessing phase.

However, our algorithm does not know this particular $r$
before-hand. So, it tries all the vertices in $\RR_{k}$ and
finds the required $r$.
The running time of the above algorithm for a fixed $t$
and a $k$-far edge  is $\TL(\frac{\NS}{2^k})$.
Since there can be at most  $2^k \NBS \log n$ $k$-far edges on $st$ path,
the total time taken to find the replacement path for all $k$-far edges for a fixed $t$ is $\TL (n)$. Since $k \le \log \NS$, the total time taken to find replacement path for all far edges in $st$ path is $\TL( n)$.  Thus, we can
find replacement path for each far edge in $st$ path for each $s \in \SSS$ and $t \in V$ in $\TL(\SI n^2)$ time.


  \section{Near Edges}
  \label{sec:near}

  There can be two types of replacement path that avoids a near edge $e$ on a $st$ path where $s \in \SSS$ and $t \in V$.
  \begin{enumerate}
  \item (Small replacement path) $|st \diamond e| \le |se| +  2\NBS \log n $
  \item  (Large replacement path) $|st \diamond e| > |se| +2   \NBS \log n$
  \end{enumerate}
  We say that first set have small replacement paths avoiding a near edge, while the second set of paths have large replacement paths avoiding a near edge.

  \subsection{Small Replacement Paths avoiding a near edge}
  \label{sec:smallsuffix}
  In this section, we will  find all  small replacement
  paths from $s$ to $t$  that avoid a near edge.
  To this end, we will make an auxiliary graph $G_s$. This graph will encode the shortest path from $s$ to other vertices $t \in V$ avoiding near edges on $st$ path. After making this graph, we will run Dijkstra's algorithm on it. At the end of this section, we will show that the output of  Dijkstra's algorithm will give us all small replacement paths.

  {\bf Construction of the auxiliary graph}: The graph  $G_s$ contains a single source node $[s]$. For each $t \in V$, there is a node $[t]$ in $G_s$. For each near edge $e \in st$ path, there is a node $[t,e]$ in $G_s$.
  We will now add edges in this graph. There is an edge from $[s]$ to $[v]$ with weight $|sv|$ for each $v \in V$.
  There is an edge from $[v]$ to $[t,e]$ of weight $1$ if $e$ does not lie in $sv$ path
  and $v$ is a neighbor of $t$.
  For each $[v,e]$, there is an edge from $[v,e]$ to $[t,e]$ with weight $1$ if $t$ is a neighbor of $v$.

  {\bf Size of the auxiliary graph}: Let us first find the number of vertices in $G_s$. For each $t \in V$, there is a node $[t]$ and $[t,e]$ where $e$ is a near edge on $st$ path. Thus, for each $t \in V$, we add $\TL (\NBS)$ vertices in $G_{s}$. Thus, the total number of vertices in $G_s$ is $\TL(n \NBS)$. Let us now calculate the number of edges in $G_{s}$. There may be an edge from $[s]$ to every other node in $G_s$. For each $v \in V$, there may be an edge from $[v]$ to $[t,e]$ where $t$ is a neighbor of $v$. But there are only $\TL (\NBS)$ vertices of type $[t,\cdot]$. This implies that the total number of edges of $[v]$ is $\TL(\deg(v) \NBS)$. Similarly, there are at most $\TL(\deg(v))$ edges out of node $[v,e]$. This implies that the total number of edges in $G_s$ is $\TL(n \NBS +     \sum_{v \in V} \deg(v) \NBS + \sum_{[v,e] \in G_s} \deg(v)) = \TL(m \NBS +n \NBS )$.

  {\bf Time taken to construct the auxiliary graph}: Let us now try to find the time taken to construct the graph $G_s$. We can find all near edges on $st$
  path in $\TL(\NBS)$ time using $\TT_s$. Thus, creating the nodes in the graph takes $\TL(n \NBS)$ time. Let us
  now find the time taken to add an edge in the graph. For each $[v]$, we need to add an edge
  from $[v]$ to $[t,e]$ if $e \notin sv$ path and $t$ is a neighbor of $v$. We can check if $e$ lies in $sv$ path
  by using $\LCA$ query in $\TT_s$. Thus, adding the edge takes $O(1)$ time. Similarly adding
  an edge out of $[v,e]$ also takes $O(1)$ time. Thus, the time taken to make $G_s$ is proportional
  to the number of vertices and edges in $G_s$.

  {\bf Time to run Dijkstra's algorithm in the auxiliary graph}: We now run Dijkstra's algorithm in $G_s$. Let $w[t,e]$ be the weight of the path
  from $[s]$ to $[t,e]$ returned by Dijkstra's algorithm. We then set $d(s,t,e) \leftarrow \min\{d(s,t,e), w[t,e]\}$. The time taken to run Dijkstra's algorithm in $G_s$ is $\TL(m\NBS + n\NBS)$. Thus, the total time taken to
  construct all $\SI$ auxiliary graphs and run Dijkstra's algorithm in them is $\TL\ (m \NS\ +\ \SI n^2)$ time.

  {\bf Proof of Correctness}: We are now ready to prove the correctness of our algorithm.
  To this end, we show the following:

  \begin{lemma}
  \label{lem:alwaysset}
  Fix a $t \in V$ and $s \in \SSS$. Let $P$ be a replacement path avoiding a near edge  $e$  on $st $ path.
  If $ |P| \le|se| + 2 \NBS \log n$,
  Then our algorithm sets  $d(s,t,e)$ to $|P|$.
  \end{lemma}

  \begin{proof}
  \iflong
  We will prove the statement using induction on edge length
  of the replacement path. That is, we will prove the statement
  for all the replacement paths of edge length 0, then edge length 1
  and so on.
  The base case is trivial, there is only one replacement
  path of edge length 0, that is $d(s,s,\emptyset) = 0$. Using the induction hypothesis, let us assume that the statement is
  true for all replacement paths of edge length $i-1$. Let us assume
  that $P$ contains $i$ edges and satisfies the condition of the lemma.
   And the last edge on this path
  is $(v,t)$.
  There are three cases:

  \begin{enumerate}
  \item $e$ does not lie on $sv$ path.

  If $e$ does not lie on $sv$ path, then $d(s,v,e) = d(s,v)$.
  But even in the auxiliary graph $G_s$, there is a path from $[s] \rightarrow [v] \rightarrow [t,e]$. The weight of this path is $|sv| + 1$. Thus, even our algorithm will
  set $d(s,t,e) = |sv|\ + 1$.

  \item $e$ lies on $sv$ path and $e$ is a far edge on $sv$ path.

  We will show that this case cannot arise. If $e$ is a far edge on $sv$ path, then $|sv|= |se| + |ev| > |se| +\ 2 \NBS \log n$. Thus the replacement path from $s$ to $v$ avoiding $e$, that is $P \setminus (v,t)$, has weight   $>|se| + 2 \NBS\ \log n$. This implies $|P| > |se| + 2 \NBS \log n$. This contradicts the statement of the lemma, namely $|P| \le |se| +\ 2 \NBS \log n$.

  \item $e$ lies on $sv$ path, $e$ is a near edge on $sv$ path and $P \setminus (v,t)$ is a
  large replacement path.

  Even in this case, the weight of $P \setminus (v,t)$ is $ > |se| + 2 \NBS \log n$.
  So similar to above, this case cannot arise.

  \item $e$ lies on $sv$ path, $e$ is a near edge on $sv$ path and $P \setminus (v,t)$ is
  a small replacement path.

  The replacement path from $s$ to $v$ avoiding $e$ is $P \setminus (v,t)$. This path has $i-1$ edges. Using induction hypothesis, we have set $d(s,v,e)$ correctly. Thus, Dijkstra's algorithm will set $d(s,t,e) = d(s,v,e) +|vt|$.

  \end{enumerate}
  \fi
  \end{proof}

  \subsection{Large Replacement Paths avoiding a near edge}
  Let $P$ be a replacement path from $s$ to $t$ avoiding $e$
  such that $|st \diamond e| > |se| +2   \NBS \log n$. We will first prove a simple observation:

  \begin{lemma}
  \label{lem:longsuffix}
  Let $P$ be a replacement path from $s$ to $t$ avoiding a near edge $e$
  such that $|P| > |se| +2   \NBS \log n$. Then $|\SUF(P)| > 2 \NBS \log n$.

  \end{lemma}
  \begin{proof}

  Remember that the suffix of $P$ will start from a vertex before $e$ on $st$ path. Let this vertex be $z$. Then $|sz| \le |se|$. Also, $|P| = |sz| + |\SUF(P)|$. But $|P|\ >\ |se| +\ 2 \NBS \log n$. This implies that $|sz| + |\SUF(P)| > |se| +\ 2 \NBS \log n$. Since $|sz | \le |se|$, it follows that $|\SUF(P)| > 2 \NBS \log n$.

  \end{proof}

  Since  $\SUF(P) > 2 \NBS \log n$,  with a high probability, there exists a
  landmark vertex $r \in \RR_0$ such that the distance of $r$
  to $t$ on $\SUF(P)$ is $\le \NBS \log n$.
  The  proof for this is similar to Lemma \ref{lem:rexists}.
  We state this lemma without proof.

  \begin{lemma}
    \label{lem:rexists1}
    Let $\PP$ be the set of all large replacement paths from $s \in \SSS$ to $t \in V$ that avoid a near edge. Given any path $P \in
    \PP$ such that $P$ avoids a near edge $e$ on $st$ path, with a high probability there exists a vertex
     $r \in \RR_{0}$ on $\SUF(P)$ such that the distance of $r$
    to $t$ on $\SUF(P)$ is $\le \NBS \log n$.
  \end{lemma}
  Now, we will  find the $r$ stated in the above lemma. Once we find
  this $r$, we can calculate $d(s,t,e)$ as follows: $d(s,t,e)
  = d(s,r,e) +\ d(r,t,e)$.
  We would have liked to write $d(r,t)$ instead of $d(r,t,e)$
  as we have not calculated $d(r,t,e)$ beforehand.
  In the following lemma, we will show that
  $e \notin rt$, implying that $d(r,t,e) = d(r,t)$.

  \begin{lemma}
  \label{lem:longP}
  Let $P$ be the shortest replacement path from $s$ to $t$ avoiding a near edge $e$ on $st$ path such that $|P|  > |se| +\ 2 \NBS \log n$. Then  there exists a landmark vertex $r$ on $\SUF(P)$ such that $st \diamond e = sr\diamond e  + rt$ and $e \notin rt$.
  \end{lemma}
  \begin{proof}
  By Lemma \ref{lem:longsuffix}, $|\SUF(P)| > 2 \NBS\ \log n$. Using Lemma \ref{lem:rexists1},
  there exists a
  landmark vertex $r \in \RR_0$ such that the distance of $r$
  to $t$ on $\SUF(P)$ is $\le \NBS \log n$.
  Thus, $st \diamond e = sr \diamond e + rt \diamond e$.
  Also we claim that $|rt \diamond e| \le \NBS \log n$, since the distance from $r$ to $t$ on
  $\SUF(P)$ is $\le \NBS \log n$. We will now show that $rt \diamond e = rt$, that is $e$ does not lie on $rt$ path.
  \begin{figure}[hpt!]
  \begin{center}
      \begin{tikzpicture}
      \draw [thick](0,0) -- (8, 0);
      \draw (0,0) node{$\bullet$};
      \draw (0,0) node[above]{$s$};
      \draw (8,0) node{$\bullet$};
      \draw (8,0) node[above]{$t$};
      \draw (5.5,0) node {\Cross} node[above]{$e$};
      \draw (5,0) node{$\bullet$};
      \draw (5,0) node[above]{$u$};
      \draw (6,0) node{$\bullet$};
      \draw (6,0) node[above]{$v$};
      \draw [dashed](2.5,0) ..controls (5, -1.5).. (7, -1.5);
      \draw [thick,dashed](5,0) ..controls (5, -1).. (7, -1.5);
      \draw (7,-1.5) node{$\bullet$};
      \draw (7,-1.5) node[below]{$r$};
          \draw [dashed](7,-1.5) ..controls (7.5, -1.2)..
  (8,0);
          \draw(8.55, -1) node{$\leq  \NBS  \log n$};
          \draw(5.55, -1.5) node[below]{};
      \end{tikzpicture}\\
  \end{center}
  \caption{Alternate path avoiding $e$, $|P'| = su + ur + rt \diamond e$}
  \end{figure}
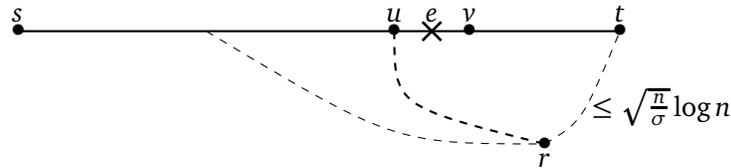
  Assume for contradiction that  $e$ lies on the $rt$ path.  Let $e=(u,v)$. Since $e$ lies on both $st$
   and $rt$ path, there exists a path $su$ and $ru$.  Since
  $rt \diamond e$ path itself is of length $ \le \NBS \log n$, $|ru|
  \le \NBS \log n$. Consider the following replacement path avoiding  $e$, $P' = su\ + ur + rt \diamond e$. Thus, $|P'| \le  |se| +2 \NBS \log n$.
  This contradicts the condition of the lemma, which says that the shortest replacement path
  avoiding $e$, that is $|P| > |se| +
   2\NBS \log n$.

  \end{proof}

  Thus, we need to process only those $r \in \RR$ such that $e \notin rt$. Our algorithm, thus, is very simple.

  \begin{algorithm}[H]
  \caption{Algorithm for finding large replacement paths avoiding a near edge}

  \ForEach{$r \in \RR_0$}
  {

     \If{ $e$ does not lie on $rt$ path}
     {
         $d(s,t,e) \leftarrow \min\{ d(s,r,e) + d(r,t), d(s,t,e)\}$\;
     }
  }
  \end{algorithm}

  The \textsc{If} condition in the above algorithm can
  be checked in $O(1)$ time using $\LCA$ queries in $\TT_r$. Also, $d(s,r,e)$ and $d(r,t)$
  are calculated during pre-processing and can be queried
  in $O(1)$ time using the hash-table. Thus, the running time
  of the above algorithm for a fixed $e$ and $t$ is $\tilde
  O(\NS)$. Since there are at most $\tilde O(\NBS)$
  near edges in $st$ path, finding all replacement paths takes $\TL(n)$ time. Thus, finding
  all large replacement paths for all vertices and for all sources takes $\tilde O(\SI n^2)$ time.

  For the $\SSR$ problem, our algorithm is now complete.
  We can use \cite{MalikMG89,Hershberger2001,NardelliPW03} to find all replacement paths
  from $s$ to vertices in $\RR$ in $\TL(m \sqrt n)$ time.
  The reader can check that the total running time taken by our algorithm in Section  \ref{sec:far} and \ref{sec:near}
  is $\TL(m \NS + \SI n^2)$ which is $\TL(m \sqrt n + n^2)$ when $\SI = 1$. Thus, we claim the following
  theorem:

  \begin{theorem}
  There is a randomized algorithm that solves the $\SSR$ problem in $\TL(m \sqrt n + n^2)$ time.
  \end{theorem}

  In the rest of the paper, we will generalize the result to multiple sources.

  \section{Generalizing to Multiple Sources}
  \label{sec:general}

  When there is one source $s$, we can find the replacement path from $s$ to all vertices in $\RR$ using the algorithm
  of \cite{MalikMG89, Hershberger2001,NardelliPW03}. However, we cannot
  use this algorithm when there are multiple sources as it
  leads to the running time of $\TL((m+n)\SI\NS )$. In this
  section, we describe a procedure that will find all replacement
  paths from $s \in \SSS$ to $r \in \RR$ in $\TL(m\NS\ +\
  \SI n^2)$ time.
  Some lemmas, definitions, and results in this section
  can be seen as the   generalization of the result by Bernstein
  and Karger \cite{Bernstein2009}.

  We sample another set of vertices which we call as centers (to differentiate them from  landmark vertices).
  Let $\CC_k$ be the set of centers sampled with the probability $\frac{4}{2^k}\sqrt {\frac{\SI}{n}}$ where $0 \le k \le \log \NS$.
  Thus (similar to Lemma \ref{lem:sizeR}), with a high probability, $| \CC_k| = \TL\Big(\frac{\NS}{2^k}\Big)$.
  A center is said to have priority $k$ if it lies in $\CC_k$. Additionally, we add all
  vertices of $\SSS$ in $\CC_0$.
  Like landmark vertices (similar to Lemma \ref{lem:sizeR}), the total number of centers is $\TL( \NS )$.
  We run $\BFS$ algorithm from each center $c$ and find the shortest path tree $\TT_c$.
  This takes $\TL(m \NS)$ time.

  Fix a source $s \in \SSS$ and a vertex $r \in \RR$. We can go over the path from a source $s$ to $r$ to find a center with the highest priority. We then move from $s$ to this highest priority center, finding a list of  centers with priority in ascending order. Let $c_1$ be the first center  in the $sr$ path. Then $c_2$ be the next center with a higher priority than $c_1$. This continues till we reach the highest center on the path $sr$. Then, we find the list of centers in descending order of priority. There are at most $O(\log n)$ centers thus found. Since we are just walking on the path $sr$ in this procedure, the time taken is the size of the path which is $O(n)$. Since there are $\TL (\SI \NS)$ pairs of possible $s$ and $r$, the total time taken to find the list of centers  is $\TL ( n\SI \NS  ) = \TL(\SI n^2)$ time. These centers naturally form an interval in the $sr$ path, which we define next:

  \begin{definition}(Interval on a $sr$ path)
  Let $sr$ be a path such that $s \in \SSS$ and $r \in \RR$. Assume that we find the centers $c_1,c_2, \dots, c_{\ell}$ on this path, then we say that the path can be divided into intervals $sc_1, c_1c_2, \dots , c_{\ell}r$.

  \end{definition}

  Note that we have to find the replacement path from a
  source $s$ to a landmark vertex $r$ avoiding an edge $e$.
  To this end, we first find the pair of centers $c_1$ and
  $c_2$ between which $e$ lies in $sr$ path. The replacement
  path can be of the following three types:
  \begin{itemize}
  \item It passes through $c_1$.
  \item It passes through $c_2$.
  \item It avoids the interval $c_1c_2$.
  \end{itemize}

  The above observation is named path cover lemma in \cite{Bernstein2009}.

  \begin{lemma} (Path Cover Lemma)
  \label{lem:pathcover}
  Given a source $s$ and a landmark vertex $r$, for any edge
  $e$ on the sr path,
  let $c_1c_2$ be the centers between which $e$ lies in $sr$
  path. Then

  $sr \diamond e = \min\begin{cases}
  sc_1 + c_1r \diamond e,\\
  sc_2 \diamond e + c_2r,\\
  sr \diamond [c_1c_2]
  \end{cases}$

  where the last distance represents the shortest path from
  $s$ to $r$ avoiding  the interval $c_1c_2$.
  \end{lemma}
   The non-trivial part of the first and the second term in
  $sr \diamond e$ is $c_1r \diamond e$ and $sc_2 \diamond e$. In the first term, we want to find a replacement path from a center to a landmark vertex and in the second term we want to find a replacement path from a source to a center.

  As in \cite{Bernstein2009}, we club together first two terms of the path cover lemma.

  \begin{definition} (\MTC, Minimum through centers)
  Given any source $s$ and a vertex $r \in \RR$. For any edge
  $e$ on the sr path,
  let $c_1c_2$ be the centers between which $e$ lies in $sr$
  path. Then

  $\MTC(s,r,e) = \min\begin{cases}
  sc_1 + c_1r \diamond e,\\
  sc_2 \diamond e + c_2r
  \end{cases}$

  \end{definition}
  Thus, $\MTC$, minimum through centers, defines first two terms in the path cover lemma.
  We first calculate the $\MTC$ term. To calculate the $\MTC$ term, we have to find a replacement path from a center to a landmark vertex and  find a replacement path from a source to a center.

  Let us do the second part first.
 
  
  \subsection{Finding the replacement path from a source to a center}
	
  \label{sec:centertosource}
  The second term in the $\MTC$ mandates us to find the replacement path from each source $s $ to each
  center $c$. However, we need not find this replacement path for each edge on $sc$ path. We need
  find the replacement path only for those edges that lie in the interval ending at $c$ on the $sc$
  path. To formalize this, let us show the following lemma:


  \begin{lemma}
  \label{lem:priork}
  Let $xy$ be an interval in $sr$ path. Assume that the priority of $x$ is $k$ and the priority of $y$ is greater than priority of $x$.\ Then $|xy| = \TL\Big(2^{k} \sqrt{\frac{n}{\SI}}\Big)$.
  \end{lemma}
  \begin{proof}
  \iflong
  We will show that a path $P$ from $x$ to $s$ or $r$ on the $sr$ path of length $\ge 2^{k+1} \sqrt{\frac{n}{\SI}}\log n$ must have a vertex in $\CC_{k+1}$. This will imply that $|xy| = \TL( 2^{k} \sqrt{\frac{n}{\SI}})$.

  Consider the subpath $P[x,r]$. The probability that none of the first $2^{k+1} \NBS$ vertices in $P[x,r]$ have priority $k+1$ is $ (1-\frac{4}{2^{k}}\sqrt{\frac{\SI}{n}})^ {2^{k+1} \sqrt{\frac{n}{\SI}}\log n} \le \frac{1}{n^8}$. Since there are polynomial numbers of intervals and the centers are chosen independently of landmark vertices,  the statement of the lemma is true for every possible interval
  with high probability using the union bound.
  \fi
  \end{proof}

  If $c$ has priority $k$, then by the above Lemma \ref{lem:priork}, we just need to find the replacement path for the first $\TL(2^{k} \NBS)$ edges on $cs$ path. This is because we are sure that any interval
  ending at $c$ will contain $\TL(2^k \NBS)$ edges if $c$ has priority $k$. In
  the ensuing discussion, we will assume that we are finding the replacement path
  from $c$ to $s$ for the first $\TL(2^k \NBS)$ edges on $cs$ path.

  If $P$ is a small replacement path from $s$ to $c$
  avoiding a near edge, then we have already found it in Section \ref{sec:smallsuffix}. Thus, our
  aim will be to find following replacement paths.
  \begin{enumerate}
  \item $e$ is a far edge on $sc$ path.
  \item $e$ is a near edge but $|P| > |se| + 2 \NBS \log n$.
  \end{enumerate}

  We now show an important result which binds these replacement paths.
  \begin{lemma}
  \label{lem:rem1}
  Let $P$ be a replacement path from $s$ to $c \in \CC$ avoiding
  $e$. (a) If $e$ is an $i$-far edge on $sc$ path, then there
  exists a vertex $c' \in \CC_i$ on $\SUF(P)$ such that $|c'c| \le 2^i \NBS \log n$ and $e  \notin c'c$ (b)
  If $e$ is a near edge and $|P| > |se| + 2 \NBS \log n$,  then there exists a vertex
  $c' \in \CC_0$ on $\SUF(P)$ such that $|c'c| \le  \NBS \log n$ and  $e \notin c'c$.
  \end{lemma}

  \begin{proof}
  \iflong

  \begin{enumerate}

  \item[(a)]
  Let us first consider the case when $e$ is an $i$-far
  edge. We can now apply Lemma \ref{lem:rexists}. Even though
  Lemma \ref{lem:rexists} is proved for landmark vertices,
  one can see that it can be used even using centers as sampling
  probability of both these sets are same. Thus, there exists a $c' \in \CC_i$ in
  $\SUF(P)$ such that $|c'c| \le 2^i \NBS \log n$. Since $e$ is
  an $i$-far edge, $e$ cannot lie in $c'c$ path.

  \item[(b)] Let us now consider
  the case when $e$ is a near edge and $|P| > |se| + 2 \NBS
  \log n$. By Lemma  \ref{lem:rexists1}, there exists a $c'  \in
  \SUF(P)$ such that $|c'c| \le  \NBS \log n$ and by Lemma \ref{lem:longP}, $e \notin c'c$. Again note that
  Lemma \ref{lem:rexists1} and  \ref{lem:longP} is proved for landmark vertices,
  one can see that it can be used even for centers in $\CC_0$.
  \end{enumerate}
  \fi

  \end{proof}

  We will make an auxiliary graph $G_s$ that will find all the required replacement paths from $s$ to each center.
  This graph will encode
  replacement path from $s$ to every center. After making this
  graph, we will run Dijkstra's algorithm on it. At the end of this section, we will show that the
  output of Dijkstra's algorithm will give us all required replacement paths. 

\iflong
   {\bf Construction of the auxiliary graph:}
  The graph $G_s$ contains a single source node $[s]$. For each center $c$, we add a node $[c]$ in $G_s$.
  For each $c \in \CC_k$, we will add an $\TL(2^{k}\NBS)$ nodes in $G_s$. These are $[c,e_1], [c,e_2], \dots,
  [c,e_{\ell 2^{k}\NBS \log n}]$ representing the first $\TL(2^k \NBS)$ edges on $cs$ path
  (where $\ell$ is a suitably chosen high constant).
  Let us now add edges in $G_s$. For each $[c,e]$, if $e$ is a near edge on $sc$ path having
  small replacement path, then we would have already found it in Section \ref{sec:smallsuffix}.
  If $w[c,e]$ is the weight of this small replacement path, then we add an edge from $[s]$ to $[c,e]$ with
  weight $w[c,e]$. Else for $[c,e]$, the replacement path $P$ satisfies
  the condition of Lemma \ref{lem:rem1}. To account for these replacement paths,
  we add an edge from $[c']$ to $[c,e]$ if $e \notin sc'$ and $e \notin c'c$.
  Further, we add an edge from $[c',e]$ (if it exists) to $[c,e]$ if $e \notin c'c$.

  {\bf Size of the auxiliary graph:}
  We now calculate the number of vertices in $G_s$. For each center of priority $k$,
  there are $\TL(2^k \NBS)$ nodes in $G_s$. Thus, the total number of nodes is $\sum_{k=1}^{\log n} \frac{\NS}{2^k} 2^{k} \NBS  = \TL(n)$. Let us now count the number of edges in $G_s$. In the worst case,
  there can be an edge between any two nodes in $G_s$. Thus, the number of edges in $G_s$
  is $\TL(n^2)$.

   {\bf Time taken to construct the auxiliary graph:}
  For each center $c$ of priority $k$, we add $\TL(2^k \NBS)$ vertices in $G_s$.
  This can be found by moving from $c$ to $s$ in the tree $\TT_s$.
  For all centers, this takes $\sum_{k=1} ^{\log n} \frac{\NS}{2^k} 2^k \NBS = \TL(n)$ time.
  Thus, each vertex can be added in $O(1)$ time in $G_s$.
  We will see that we can check if an edge can be added between any pair of node in $G_s$ in $O(1)$
  time. To add an edge from $[c']$ to $[c,e]$, we need to check if $e \notin sc'$ and
  $e \notin c'c$. This can be done using $\LCA$ queries in $\TT_s$ and $\TT_c$.
  Similarly, an edge from $[c',e]$ to $[c,e]$ can be added in $O(1)$ time.
  Thus, the time taken to construct $G_s$ is $\TL(n^2)$ in the worst case.

   {\bf Time taken to run Dijkstra's algorithm in the auxiliary graph:}
  We now run Dijkstra's algorithm in $G_s$. For each $[c,e] \in G_s$, we set $d(c,s,e)$
  to be the weight of the path from $[s]$ to $[c,e]$ as returned by  Dijkstra's algorithm.
   Since, there are $\TL(n)$ vertices
  and $\TL(n^2)$ edges in $G_s$, the total time to run Dijkstra's algorithm in $G_s$ is $\TL(n^2)$.
  Since we run the above algorithm for $\SI$ sources, the total time
  taken by the above procedure is $\TL(\SI n^2)$.

   {\bf Proof of Correctness:}
  We now prove the correctness of the above procedure.

  \begin{lemma}
  Let $c$ be a center of priority $k$.  Let $P$  be a replacement path from $s \in \SSS$ to $c \in \CC$ avoiding
  $e$ such that $e$ is one of first $\ell 2^{k} \NBS \log n$ edges in $cs$ path (where $\ell \ge 2$ is a suitably chosen constant).
  Then Dijkstra's algorithm in $G_s$ will set $d(c,s,e)$ to   $|P|$.
  \end{lemma}
	
  \begin{proof}

  We will prove the lemma by induction on the length of the path. Since $s$ is also a center in $\CC_0$, the base case is that our algorithm finds a path of length 0 from $s$ to itself.  Assume that there is a replacement path $P$ from $s$ to $c$ avoiding $e$ that satisfies the statement of the lemma. If $P$ is a small replacement path for a near edge, then we have already added an edge from $[s]$ to $[c,e]$ with an appropriate weight. So, assume that $P$ satisfies the statement of the Lemma \ref{lem:rem1}. By Lemma \ref{lem:rem1}, if $e$ is an $i$-far edge, then
  there exists a vertex $c' \in \CC_i$ in $\SUF(P)$ such that $|c'c| < 2^i \NBS \log n$  and $e \notin c'c $.
  We will first show that a node corresponding to  $c'$ always exists in $G_s$.
  We claim that either $sc'$ does not pass through $e$ or $e$ is one of the first $\ell 2^{i} \NBS \log n$
  edges in $c's$ path. This is because, $|c'c| < 2^i \NBS \log n$ and $|ec| \le 2^{i+1} \NBS \log n$,
  then by triangle inequality, if $e$ lies on $sc'$ path, then $|ec'| \le 2^{i+2} \NBS \log n$.
  Thus, the node $[c',e]$ will be added in $G_s$.

  Similarly, if $e$ is a near edge
  then there exists a vertex
  $c' \in \CC_{0}$ such that $e \notin c'c$. One can argue as above
  that either $e$ does not lie in $sc'$ or $[c',e]$ vertex exists in $G_s$.

  Thus, there are  following cases:

  \begin{enumerate}
  \item $e \notin sc'$

  In this case, we have $P = sc' + c'c$. In $G_s$, we have an edge from $[s]$ to $[c']$ with weight $|sc'|$ and  $[c']$ to $[c,e]$ with weight $|c'c|$ if $c'c$ does not pass through $e$.
  Thus, the Dijkstra's algorithm will be able to find this path.

  \item $e \in sc'$

  We have already shown that $[c',e]$ will exist in $G_s$.
  $P[s,c']$ is a replacement path whose edge length is strictly less
  than $|P|$. By induction hypothesis, we assume that Dijkstra's algorithm finds
  $P[s,c']$ correctly. Thus, it sets $d(c,s,e)$ to $|P|$  correctly.
  \end{enumerate}

  \end{proof}
	\fi

  Till now, we have found all replacement path that will be used in the
  second term of $\MTC$.
  We will now try to find replacement paths
  that will be used in the first term of $\MTC$.
  
  \subsection{Finding the replacement path from a center to a landmark vertex}
  \label{sec:centertolandmark}
  Now, we calculate the first term of $\MTC$.
  Let $P$ be a replacement path from a source $s$ to a landmark vertex $r$
  that is passing through the center $c$. To calculate the first term
  in $\MTC(s,r,e)$, we should calculate $|sc| + |cr \diamond e|$ if $e$ lies in the interval
  starting with $c$ on $sr$ path. This implies that
  the center $c$ lies in $sr$ path and $e$ lies in the interval starting with $c$ in $cr$ path.

  If  $c$ has priority $k$, then by Lemma \ref{lem:priork}, on any $sr$ path, we are sure to find a center of priority $k+1$ at a distance of $\TL(2^k \NBS)$ from $c$.
  Thus, we just need to find the replacement path for edges till a distance of $\TL(2^k \NBS)$ from $c$.

  Before moving ahead, let us first make an important observation. We need to find a replacement path from $c$
  to $r$ avoiding $e$ only if there exists a replacement path from some source $s$ to $r$ avoiding $e$ that passes
  through $c$. Otherwise, there is no need to even find a replacement path from $c$ to $r$ avoiding $e$.

  Remember that a replacement path from a source to $r$ avoiding $e$ and passing through $c$
  can be of three types.

  \begin{enumerate}

    \item Small replacement path that avoids a near edge.
    \item Large replacement path that avoids a near edge.
    \item Replacement path that avoids a far edge.
  \end{enumerate}

  Let us look at the first set of replacement paths as we have already found these paths in Section \ref{sec:smallsuffix}.

  \subsubsection{Small replacement paths avoiding a near edge}
  \label{sec:generalsmall}
  For each $r \in \RR$, we have already found the replacement path from each source to $r$. We can enumerate all the edges on each of these replacement paths too.
  Remember that the algorithm in Section \ref{sec:smallsuffix} only finds the length of the small replacement path, not the replacement path itself.
  However, we can use Dijkstra's
  algorithm to find the actual path too. The time taken to enumerate a path is equal to the length of
  the path. Since, there are $\SI$ sources, $\TL(\NS)$ vertices in $\RR$, and only $\TL(\NBS)$ near edges, we need to enumerate $\TL(\SI \NS \NBS) = \TL(\SI n)$ paths. Since there can be $n$ edges on each path, the total
  time to enumerate all the paths is $\TL( \SI n^2)$.

  We can pre-process each enumerated replacement path to find whether a vertex lies on it -- this can be done using LCA queries. For a landmark vertex $r$ and a center $c$, we can check if there exists a small replacement path  from a source to  $r$ passing through $c$ and avoiding a near edge $e$. To this end, we will first check all enumeration that represents a replacement path from a source to $r$ avoiding $e$. There are $\SI$ such enumerations. If $c$ lies in any of the enumerations, then there is a replacement path avoiding a near edge of small length passing through $c$. Thus, we can find a small replacement path from
  $c$ to $r$ avoiding a near edge $e$ in $\TL(\SI)$ time. Once again we reiterate, that we will find this replacement path only if there is a replacement path from a source to $r$ avoiding $e$ that passes through
  $c$. If there is no such path, then there is no need to find a small replacement path from $c$ to $r$ avoiding
  a near edge $e$.

  Given a $(c,r,e)$, the time taken to find a small replacement path is $O(\SI)$. Since there
  are $\TL(\NS)$ centers, $\TL(\NS)$ landmark vertices and $\TL(\NBS)$ near edges, the total
  time taken to find all small replacement paths for all possible $(c,r,e)$ tuples is
  $\TL( \NS \NS \NBS \SI) = \TL( \SI n^2)$. We store the length of all small replacement paths in $d(c,r,e)$
  (where $c \in \CC$ and $r \in \RR$ and $e$ is a near edge on $cr$ path) in a hash-table for efficient retrieval.

  \subsubsection{Other replacement paths}
  \label{sec:otherreplacementpaths}

  Once we have dealt with small replacement paths, two other types of replacement
  paths are left. Remember that this replacement path $P$ is from a source $s$ to a landmark vertex $r$
  that is passing through the center $c$. Also, $e$ lies in $cr$ path. And our aim is to
  find the length of the sub-path $P[c,r]$.

  The replacement path $P$ can be of two types:

  \begin{enumerate}
  \item $P$ avoids a far edge $e$ on $sr$ path and passes through $c$.

  If $e$ is a far edge on $sr$ path then
  it is also a far edge on $cr$ path.

  \item $e$ is a near edge but $|P| > |se| + 2 \NBS \log n$.

  Since $e$ is a near edge for $sr$ path, it is also a near edge
  in $cr$ path. Thus, even the subpath $P[c,r]$ satisfies,
  $|P[c,r]| > |ce| + 2 \NBS \log n$.
  \end{enumerate}

  Thus, we need to find a replacement path $P[c,r]$ such that:

  \begin{enumerate}
  \item $P[c,r]$ avoids a far edge $e$ on $cr$ path.
  \item $e$ is a near edge but $|P[c,r]| > |ce| + 2 \NBS \log n$.
  \end{enumerate}


  \begin{lemma}
  \label{lem:rem}
  Let $c$ lie on $sr$ path such that $e \in cr$. Let $P$ be a replacement path from $s$ to $r$ avoiding $e$  and passing through $c$.
  Let $P[c,r]$ be the corresponding replacement path from $c$ to $r$  avoiding $e$ such that (a) $e$ is a far edge on $cr$ path or (b)
  $e$ is a near edge and $|P[c,r]| > |ce| + 2 \NBS \log n$. Then
  (1) $|\SUF(P[c,r])| > 2 \NBS \log n$ and (2) There exists a vertex
  $r' \in \RR$  in $\SUF(P[c,r])$ such that $e \notin r'r$.
  \end{lemma}

  \begin{proof}
  \iflong
  \begin{enumerate}
  \item If $e$ is a far edge, $P[c,r]$ has to diverge from $cr$
  before $e$. Thus, $\SUF(P[c,r])$
  will be $> 2 \NBS \log n$. And by Lemma \ref{lem:longP},
  even for
  a near edge, if $|P[c,r]| > |ce| + 2 \NBS \log n$, then $|\SUF(P[c,r])|
  > 2\NBS \log n$.

  \item  Let us first consider the case when $e$ is a far
  edge. By Lemma \ref{lem:rexists}, there exists a $r'$ in
  $\SUF(P[c,r])$ such that $|r'r| \le \NBS \log n$. Since $e$ is
  a far edge, $e$ cannot lie in $r'r$ path. Let us now consider
  the case when $e$ is a near edge nd $|P[c,r]| > |ce| + 2 \NBS
  \log n$. By Lemma \ref{lem:rexists1} and \ref{lem:longP}, there exists a $r'  \in
  \SUF(P[c,r])  $ such that $e \notin r'r$.
  \end{enumerate}
  \fi
  \end{proof}

  With all armoury at hand, we are now ready to find the required replacement
  paths from a center to all landmark vertices. Fix a center $c$ with priority $k$. We will now find the replacement path from $c$ to each landmark vertex.
  Also remember that
   we want to find replacement path avoiding all  edges at a distance $\TL(\sqrt{\frac{ n }{ \SI}}2^k)$ from $c$. 
   We will now create an auxiliary graph that will help us in finding all
   the required replacement paths.  

  \iflong
  {\bf Construction of the auxiliary graph: }
  We make an auxiliary directed weighted graph $G_c$ with a single source node $[c]$. In this graph, for each landmark vertex $r$, we will have at most  $\TL(2^k\NBS)$ nodes, $[r,e_1], [r,e_2], \dots, [r,e_{\ell2^k\NBS \log n}]$ (where $\ell \ge 2$ is a suitably chosen constant). The second term in the tuple represents the first $\TL(2^k\NBS)$ edges  on the $cr$ path. Also, there is a  node $[r]$ for each landmark vertex $r$.
  We now add edges in $G_c$.  There is an edge from $[c]$ to  $[r]$  (where $r \in \RR$) of weight $|cr|$. For each $[r,e]$, there are three types of incoming edges to it.

  \begin{enumerate}
  \item If there exists a small replacement path from a source to $r$ avoiding the near edge $e$ passing through $c$, then we have already found it in Section \ref{sec:generalsmall}. Let the weight
  of this path be $w[c,r,e]$.  We add an edge from $[c]$ to $[r,e]$ with the weight $w[c,r,e]$.

  \item Edge from $[r']$ (where $r' \in \RR)$ to $[r,e]$ of
  weight $|r'r|$ if $cr'$ path does not pass through $e$
   and $r'r$ path does not pass through $e$.

  \item For each $r' \in \RR$, if $[r',e]$ exists, then there
  is an edge from $[r',e]$ to $[r,e]$ of weight $|rr'|$
  if $r'r$ path does not pass through $e$.
  \end{enumerate}
     This completes the description of $G_c$.

  {\bf Size of the auxiliary graph:}
  Let us first count the number of vertices in $G_c$. For each $r \in \RR$, there is a node $[r]$ in $G_c$. Let us now count the number of nodes of type $[r,e]$ in $G_c$. Since there  are $\TL(\NS)$ vertices in $\RR$ and at most $\TL(2^k\NBS)$ nodes of tupe $[c,e]$, there are at most $\TL(\NS \NBS 2^k) = \TL(n2^k)$ vertices in the auxiliary graph.
  Let us now count the number of edges in $G_c$. There may be an edge from $[c]$ to every
  other vertices in the graph.
   Each vertex of type $[r']$ can have an edge to all other vertices in the graph. Thus the number of such edges is $\TL\ (n2^k\NS)$. Each vertex of type $[r',e]$  can have edge to at most $\TL( \NS)$ other edges of type $[r,e]$. There are at most $\TL(n2^k\NS)$ such edges. Thus, in total there are $\TL(n2^k\NS)$ edges $G_c$.

   {\bf Time taken to construct the auxiliary graph:}
   For each landmark vertex $r$, we add at most $\TL(2^k \NBS)$ vertices in $G_c$.
   This can be done by moving up from the vertex $c$ in $\TT_r$. Thus, adding a vertex in
   $G_c$ takes $O(1)$ time. We will see that even adding an edge in $G_c$ takes $O(1)$
   time. To add an edge from $[c]$ to $[r,e]$, we need to check
   if there is a small replacement path from $c$ to $r$ avoiding $e$.
   We have already found this path in Section \ref{sec:generalsmall} and can be retrived
   in $O(1)$ time from the hash-table. Similarly, to add an edge from $[r']$
   to $[r,e]$, we can check if $e \in cr'$ (using $\LCA$ query in $\TT_c$) and
   $e \in r'r$ (using $\LCA$ query in $\TT_{r'}$) in $O(1)$ time.
   Similarly, a edge from node $[r',e]$ to $[r.e]$ can be added in $O(1)$
   time. The reader can check that the time taken to construct this auxiliary graph
   is proportional to its size.

  {\bf Time taken to run Dijkstra's algorithm in the auxiliary graph: }
  We now run Dijkstra's algorithm in $G_c$. For each $[r,e] \in G_c$, we set $d(c,r,e)$
  to be the weight of the path from $[c]$ to $[r,e]$ as returned by Dijkstra's algorithm.
  We now calculate the time taken to run Dijkstra's algorithm.
  The time taken to run Dijkstra's algorithm in $G_c$ is $\TL(n 2^k \NS)$.
     Since there are $\frac{\NS}{2^k}$ centers of priority $k$ and at most $\log n$ such priorities, the total time taken to run Dijkstra's algorithm in all auxiliary graphs is $\TL(\sum_{k=0}^{\log n}  \frac{\NS}{2^k}  n\NS 2^k)  =\TL(\SI n^2)$.

  {\bf Proof of Correctness: }
  We now show that the shortest path calculated above gives us the replacement path from $c$ to all landmark vertices.
  \begin{lemma}
   Let $c$ be a center of prioroty $k$. Let $P$ be a replacement path from $s$ to $r$ avoiding $e$  and passing through $c$ such that $e$ is one of
   the first $\TL(2^k \NBS)$ edges on $cr$ path. Then Dijkstra's algorithm in $G_c$ gives the replacement path $P[c,r]$ avoiding $e$.
  \end{lemma}
  \begin{proof}

  We will prove the lemma by induction on the length of the path $P[c,r]$. The base case is that our algorithm finds a path of length 0 from $c$ to itself (for the base case we will assume that $c$ is also a landmark vertex).  Assume that there is a replacement path $P$ from $s$ to $r$ avoiding $e$ that passes through $c$. If $P$ is a small replacement path for a near edge, then we have already added an edge from $[c]$ to $[r,e]$ with an appropriate weight. So, assume that $P$ satisfies the statement of the Lemma \ref{lem:rem}. By Lemma \ref{lem:rem},  there exists a vertex
  $r' \in \RR$ such that $e \notin r'r$. We now claim that either $e \notin cr'$ or
  $e$ is one of the first $\TL(2^k \NBS )$ edges in $cr'$ path. This is because
  if $e \in cr'$ and $e \in cr$, then the subpath $ce$ is same for both of them.

  There are following cases:

  \begin{enumerate}
  \item $e \notin cr'$

  In this case, we have $P[c,r] = cr' +\ r'r$. In $G_c$, we have an edge from $[c]$ to $[r']$ with weight $|cr'|$ and  $[r']$ to $[r,e]$ with weight $|r'r|$ if $r'r$ does not pass through $e$.
  Thus, the Dijkstra's algorithm will be able to find $P[c,r]$.



  \item $e \in cr'$ 

  In this case,  the replacement path from $c$ to $r$ is $cr' \diamond e + r'r$.
  We have already shown that $[c,r']$ exists in $G_c$.
  We have added an edge from
  $[r',e]$ to $[r,e]$ with weight $|rr'|$.
  By induction hypothesis, we assume that Dijkstra's algorithm correctly calculates $cr' \diamond e$. Thus, it will correctly calculate $cr \diamond e$ too.

  \end{enumerate}
  
  \end{proof}
\else
  \fi

  Given the result in this section (Section \ref{sec:centertolandmark}) and the result in Section \ref{sec:centertosource}, we can now calculate the first two terms in the path cover lemma (See Lemma \ref{lem:pathcover}).  In the ensuing discussion, we will be calculating a replacement path from a source to
  a landmark vertex that avoids an entire interval.

  \subsection{Replacement path avoiding an interval}
  In this section, we will find the replacement path from $s$ to $r$ avoiding the interval that contains edge $e$. To this end, we use the concept of bottleneck vertex (adapted as bottleneck edge for our purpose) introduced in \cite{Bernstein2009}. Let us first define a few terms which will be used in this section.

  \begin{definition}
  Let $[s,r,i]$ denote the $i$-th interval on the $sr$ path where $s \in \SSS$ and $r \in \RR$.   A bottleneck
  edge is  the hardest edge on this interval
  to avoid. Formally, the bottleneck edge  $\BB[s,r,i] := \max_{e \in [s,r,i]} \{sr \diamond e\}$.
  \end{definition}

   Thus, the path cover lemma looks as follows for an edge that lies in the $i$-th interval of $sr$ path.

  \begin{lemma}
  \label{lem:pathcovermod}
  If $e$  lies on the $i$-th interval on the $sr$ path,
  then \\
  $sr \diamond e= \min\begin{cases}
  \MTC(s,r,e),\\
  sr \diamond \BB[s,r,i]
  \end{cases}$
  \end{lemma}
  \begin{proof}
   If $|sr \diamond e|$ avoids the $i$-th interval
  then it avoids $\BB[s,r,i]$ too.
  Thus, $|sr \diamond e| \ge |sr \diamond \BB[s,r,i]|$.
    But $\BB[s,r,i]$ is the bottleneck edge of the $i$-th interval, so
    $|sr \diamond e| \le |sr \diamond \BB[s,r,i]|$.
    This implies that    $|sr \diamond e| = |sr \diamond \BB[s,r,i]|$.
    If $|sr \diamond e|$ does not avoid the $i$-th interval, then
    $|sr \diamond e| = \MTC(s,r,e)$. This proves the statement of the lemma.

  \end{proof}

  Thus, two things are left now: find the bottleneck edge for each interval and then find the replacement path avoiding the bottleneck edge.

  \subsubsection{Finding Bottleneck edge in each interval in $sr$ path}
  \label{sec:findbottleneck}
  
  We now show how to find a bottleneck edge in the $i^{th}$ interval of the path $sr$. We first observe that the bottleneck edge will have the highest $\MTC$ value among all edges in the $i$-th interval. This is true as by Lemma \ref{lem:pathcovermod}, the last term for each edge in the interval is the same. So, to find a bottleneck edge, we should look at the edge in the interval which maximizes the first two terms, that is the $\MTC$ value.

   To find the bottleneck edge of
  the $i$-th\ interval, we just need to go over each edge in the
  $i^{th}$ interval and check the $\MTC$ value (whose constituents we have already calculated).  This takes $O(n)$
  time for all intervals on the $sr$ path. Since there
  are $\SI$ sources and $\TL(\NS)$ vertices in $\RR$, the
  total time taken to find  bottleneck edges is $\TL(n\
  \SI \NS) = \TL(\SI n^2)$.

  \subsubsection{Finding the replacement path avoiding the bottleneck edge in $sr$ path}

  Let us now find the replacement path $P$ avoiding the bottleneck edge of the $i$-th interval on the $sr$ path.
  Let $e \leftarrow \BB[s,r,i]$, that is $e$ is the bottleneck edge of the $i$-th interval in $sr$ path.
  If $P$ is a small replacement path avoiding a near edge on $sr$ path, then we would have already found it in Section \ref{sec:smallsuffix}. Thus, our focus will be to find $P$ when:
  \begin{enumerate}
   \item $P$ avoids a far bottleneck edge $e$.
   \item $P$ avoids a near bottleneck edge $e$ but $P > |se| + 2 \NBS \log n$.
  \end{enumerate}

   By Lemma \ref{lem:rem}, there exists a $r'$ in $\SUF(P)$ such that $e \notin r'r$. We will now crucially use this property to make another auxiliary graph $G_s$.

   {\bf Construction of the auxiliary graph:}
   This graph $G_s$ contains a source vertex $[s]$. There is a vertex $[r]$ for each $r \in \RR$. For a bottleneck edge of interval $i$ in $sr$ path, there is a vertex $[s,r,i]$ where $i \le \log n$.
  We now find the edges in $G_s$. There is an edge from $[s]$ to $[r]$ with weight $|sr|$.
  If $\BB[s,r,i]$ happen to be a near edge whose replacement path has small weight, then we  add an edge from $[s]$ to $[s,r,i]$ with appropriate weight (see Section \ref{sec:smallsuffix}). Else,  there can be three types of edges to $[s,r,i]$.
  \begin{enumerate}
  \item Edge from $[s]$ to $[s,r,i]$ of weight $\MTC(s,r,\BB[s,r,i])$. This edge represents the first
  two terms in the path cover lemma for the bottleneck edge $\BB[s,r,i]$ (we have already calculated these
  in Section \ref{sec:centertosource} and \ref{sec:centertolandmark}).

  \item For each $r' \in \RR$, there is an edge from $[s]$ to $[s,r,i]$ with weight $\MTC(s,r',\BB[s,r,i]) +\ |r'r|$ if $\BB[s,r,i]$ does not lie in $r'r$.

  \item If  $\BB[s,r,i]$ lies in the $j^{th}$ interval on the $sr'$ path, then there is an edge from $[s,r',j]$ to $[s,r,i]$ with weight $|r'r|$ if $\BB[s,r,i]$ does not lie in $r'r$ path.
  \end{enumerate}

   This completes the construction of $G_s$.

  {\bf Size of the auxiliary graph:}
  The number of nodes of type $[r]$ in $G_s$ is $\TL(\NS)$.
  The number of nodes in $G_s$ of type $[s,r,i]$ is $\TL(\NS)$ since there are $\TL(\NS)$ landmark vertices and $\log n$ interval in any $sr$ path. We now find the number of edges in $G_s$. For each vertex $[s,r,i]$, there are at most $\TL (\NS)$ from the source $[s]$ (due to point (2) in above enumeration). Also, there are at most $\TL(\NS)$ edges from other vertices in $\RR$ (due to point (3) in above enumeration). Thus, the total number of edges in $G_s$ is $\TL ( \NS\NS) = \TL(n^2)$.

  {\bf Time taken to construct the auxiliary graph: }
  For each source, we have already found the bottleneck edge of each interval in
  $sr$ path in Section \ref{sec:findbottleneck}(where $r$ is a landmark vertex).
  Thus, adding vertices in $G_s$ takes $O(1)$ time. If there is a small replacement
  path from $[s]$ to $[s,r,i]$, then we have already found it in Section \ref{sec:smallsuffix}
  and can be added in $O(1)$ time. We add an edge from $[s]$ to $[s,r,i]$ with
  $\MTC(s,r,\BB[s,r,i])$. Again, we have calculated the $\MTC$ term in Section \ref{sec:centertosource}
  and \ref{sec:centertolandmark}. So, we can add this edge in $O(1)$ time.
  For each $r' \in \RR$, we add an edge from $[s]$ to $[s,r,i]$ if $e \notin r'r$.
  Again, this edge can be added in $O(1)$ time. The hardest part is point (3)
  in the above enumeration. For each interval $[s,r',j]$,
  we first need to check if $\BB[s,r,i]$ in the $j$-th interval in
  $sr'$ path. This can be done by first finding if $\BB[s,r,i]$ lies
  in $sr'$ path -- by doing $\LCA$ queries in $\TT_s$. If $\BB[s,r,i]$
  lies in $sr'$ path, then we can calculate the distance of $\BB[s,r,i]$
  relative to $s$ and $r'$. This can be done easily as we have already
  stored distances from $s$ to all other vertices in the graph in $d(s.\cdot)$ (in the pre-processing phase).
  Thus, all edges in $G_s$ can be added in $O(1)$ time.
  Thus, the time taken
  to construct $G_s$ is equal to the worst case size of $G_s$, that is $\TL(n^2)$.

  {\bf Time taken to run Dijkstra's algorithm in the auxiliary graph: }
   We run Dijkstra's algorithm in $G_s$ to find the shortest replacement path
   for each bottleneck edge. We set $d(s,r,\BB[s,r,i])$ to the weight of the shortest path
   from $[s]$ to $[s,r,i]$ as returned by Dijkstra's algorithm in $G_s$. The time taken by Dijkstra's algorithm
   in $G_s$ is  $\TL( n^2)$. Since there are $\SI$ such graphs, the total time
  taken is $\TL( \SI n^2)$.

  {\bf Proof of Correctness: }
  We now prove the correctness of the above algorithm.
  \begin{lemma}
  \label{lem:last}
  Let $P$ be the shortest path from $s \in \SSS$ to $r \in \RR$ avoiding the
   bottleneck edge in the $i$-th interval  of $sr$ path. Then Dijkstra's algorithm in $G_s$ correctly
  finds $P$.
  \end{lemma}

  \begin{proof}
  \iflong
  We will prove using induction on the edge length of the replacement paths.
  Since a source is also in $\RR$, Dijkstra's algorithm correctly
  finds the replacement path of length 0 (which is our base case).
  Let us assume that the number of edges in $P$ is $k$. By induction hypothesis,
  Dijkstra's algorithm has correctly found replacement path from $s$ to any $r' \in \RR$
  avoiding a bottleneck edge on $sr'$ path whose edge length is $< k$.
  If $P$ is a small replacement path avoiding a near bottleneck edge $\BB[s,r,i]$
  in $sr$ path, then we would have already found it in Section \ref{sec:smallsuffix}
  and put an edge from $[s]$ to $[s,r,i]$ of appropriate weight in $G_s$.
  So, assume that $P$ satisfies the statement of the Lemma \ref{lem:rem1}.
  (Remember that we prove this lemma for centers, but the reader can check that
  the same lemma holds for landmark vertices as both the
  sets have same sampling probability).
  By Lemma \ref{lem:rem1},  there exists a vertex
  $r' \in \RR$ such that $\BB[s,r,i] \notin r'r$.

  There are following cases

  \begin{enumerate}
  \item $\BB[s,r,i] \notin sr'$

  In this case, we have $P = sr' +\ r'r$. In $G_s$, we have an edge from $[s]$ to $[r']$ with weight $|sr'|$ and  $[r']$ to $[s,r,\BB[s,r,i]]$ with weight $|r'r|$.
  Thus, the Dijkstra's algorithm will be able to find this path.

  \item $\BB[s,r,i] \in sr'$

  Let us assume that $\BB[s,r,i]$ lies on the $j$-th interval in $sr'$ path. Then, \\
  $d(s,r,\BB[s,r,i]) = \min \begin{cases}
  \MTC(s,r',\BB[s,r,i]) + |r'r| \\
  sr' \diamond \BB[s,r',j] + |r'r| \\
  \end{cases}$.

  Here we are just expanding the path cover lemma for the tuple $(s,r',\BB[s,r,i])$.
  For the first case, we have added an edge from $[s]$ to $[s,r,i]$ with weight $\MTC(s,r', \BB[s,r,i])$.

  And for the second term, since $sr' \diamond \BB[s,r',j]$ has edge length strictly
  less than $k$, using induction hypothesis, we can assume that we have set
  $d(s,r',\BB[s,r,j]) = |sr' \diamond \BB[s,r',j]|$. Thus, Dijkstra's algorithm will find $P$ correctly.
  \end{enumerate}
  \fi
  \end{proof}
  Thus, we claim the main theorem of the paper:
  \begin{theorem}
    There is a randomized combinatorial algorithm that solves the $\MSR$ problem in $\TL(m\NS + \SI n^2)$.
  \end{theorem}


\section{Conditional Lower Bounds}
\label{sec:lower}
We now try to prove a conditional lower bound of $\Omega(m\sqrt{n\sigma})$ for $\MSR$ problem with $\sigma$ sources in undirected and unweighted graphs by giving a combinatorial reduction
from Boolean Matrix Multiplication (\BMM) to \MSR. This can be seen as a simple extension
of the lower bound obtained in \cite{ChechikC19}.

Let $\BMM(n,m)$  be a combinatorial algorithm  for multiplying two matrices $A$ and $B$ both of size $n \times n$  such that the total number of 1's in both $A$ and $B$ is $m$. A combinatorial algorithm does not use any matrix multiplication. The conditional lower bound relies on the conjecture for combinatorial $\BMM$, that there does not exist any truly subcubic algorithm for it.

\begin{conjecture}
 In the Word RAM model with words of $O(\log n)$ bits, any combinatorial algorithm for multiplying two Boolean  matrices $A$ and $B$ of size $n \times n$ with a total number of $m$ 1's in them requires $(mn)^{1-o(1)}$ time in expectation to compute.
\end{conjecture}
Let $\MSR(n,m)$ denote our multiple source replacement path algorithm for unweighted graph with $n$ vertices, $m$ edges and $\sigma$ sources. We will now reduce $\BMM(n,m)$ to $\MSR(n,m)$

\begin{theorem}
For a combinatorial algorithm $\MSR(n,m)$ with runtime of $T(n,m)$, there is a combinatorial algorithm for $\BMM(n,m)$  problem with runtime of $O(\NBS T(O(n),O(m)))$.
\end{theorem}

\begin{proof}
Consider three matrices $A,B$ and $C$ such that $C = A \times B$. Next we show how to compute $C$ using our $\MSR(n,m)$ algorithm. To this end, we create $\NBS$ graphs $\{G_1,G_2, \dots , G_{\NBS}\}$  each containing 3 sets of vertices $V_a = \{a(1),a(2),\dots,a(n)\}, V_B = \{b(1),b(2),\dots,b(n)\}$  and $V_c = \{c(1),c(2),\dots,c(n)\}$. Each graph will have $O(n)$ vertices and $O(m)$ edges.
 In each graph $G_{i}$, we will use $\MSR$ algorithm  to find all values of rows $C[(i-1) \times \NS + j]$ for all $1\leq i \leq \NBS$,$1\leq j \leq \NS$.

Let us give the construction of $G_i$.
 For all $ 1 \le x,y \le n$, we add an edge between $a(x)$ and $b(y)$ if $A[x][y]=1$. Similarly we add an edge between $b(x)$ and $c(y)$ if $B[x][y]=1$. We add additional vertices in graph $v(1),v(2),\dots, v(\NS)$. We create $\sigma$ paths such that $P_{j}=v((j-1) \NBS +1),v((j-1) \NBS  +2),\dots,v(j\NBS)$, where $1\leq j \leq \sigma$ and each vertex $v(j\NBS)$ is a source vertex.

After that, we connect $v(j)$ to $a((i-1)\NS+j)$ by a path of $2((j-1) \mod \NBS) + 1$ additional vertices, where $1\leq j \leq \NS$. Thus, the distance of $v(1)$ from $a(1)$ is 1, distance of $v(2)$ from $a(2)$ is 3 and so on. The distance is again reset at $j = \NBS+1$, the distance of $v(\NBS+1)$ from $a(\NBS+1)$ is 1. This graph construction will give all values of row $C[(i-1)\NS+j]$, for all $1\leq j \leq \NS$ by running $\MSR$ algorithm on $G_i$.

Let us consider first graph $G_1$ and its first source $v(\NBS)$. If the shortest path from $v(\NBS)$ to $c_{\ell}$ for all $1\leq \ell \leq n$ is of length $\NBS+3$ then we set $C[1][\ell]=1$ (in this case path will be $v(\NBS),v(\NBS-1),\dots, v(1),a_1,b_{\ell'},c_{\ell})$), otherwise it is $0$. For edge failure $e(v(1),v(2))$, if there  is path length is $\NBS+5$ then $C[2][\ell]=1$ (in this case path will be $v(\NBS),v(\NBS-1),\dots,v(2),\dots,a_2,b_{\ell'},c_{\ell})$ in $G_{i}/e(v(1),v(2))$), otherwise $0$.
Similarly, one can see that we can find all values of rows $ C[1], C[2], \dots, C[\NS]$ by executing
$\MSR$ algorithm in $G_1$. Thus, by running $\MSR$  algorithm on $G_1, G_2, \dots, G_{\NBS}$, we
can find all rows of $C$.
\end{proof}

\bibliographystyle{plain}
\bibliography{paper}

\end{document}